\documentclass[12pt]{article}

\usepackage{amssymb, amsmath, amsthm}
\usepackage{graphicx}
\usepackage{cite}

\newcommand{\td}{\text{d}}

\usepackage{hyperref}

\def\be{\begin{equation}}
\def\ee{\end{equation}}
\def\bea{\begin{eqnarray}}
\def\eea{\end{eqnarray}}

\newtheorem{theorem}{Theorem}
\newtheorem{lemma}{Lemma}
\theoremstyle{definition}

\newtheorem{remark}{Remark}

\usepackage[left=2cm,right=2cm,top=2cm,bottom=2cm]{geometry}

\title{\bf No static bubbling spacetimes in higher dimensional Einstein-Maxwell theory}
\author{Hari K. Kunduri$^a$\footnote{hkkunduri@mun.ca } \  and James Lucietti$^b$\footnote{j.lucietti@ed.ac.uk } \\ \\
\small \sl $^a$ Department of Mathematics and Statistics, \\  \small \sl McMaster University,\footnote{On sabbatical leave from Department of Mathematics and Statistics, Memorial University of Newfoundland
St. JohnÕs, NL  A1C 5S7, Canada} \\ \small \sl Hamilton, ON, L8S 4K1, Canada
\\ \small \sl $^b$  School of Mathematics and Maxwell Institute of Mathematical Sciences, \\ \small \sl University of Edinburgh, \\ \small \sl   King's Buildings, Edinburgh, EH9 3FD, UK }

\date{}

\begin{document}
\maketitle
\begin{abstract} We prove that any asymptotically flat static spacetime in higher dimensional Einstein-Maxwell theory must have no magnetic field. This implies that there are no static soliton spacetimes and completes the classification of static non-extremal black holes in this theory. In particular, these results establish that there are no asymptotically flat  static spacetimes with non-trivial topology, with or without a black hole, in Einstein-Maxwell theory.
\end{abstract}

\vspace{.5cm}

A striking result in Einstein-Maxwell theory is the absence of soliton solutions, i.e., nontrivial, asymptotically flat, stationary, globally regular spacetimes \cite{Lich, Carter1972,Breitenlohner:1987dg}. Therefore, an isolated, self-gravitating, equilibrium state with positive energy must contain a black hole, a phenomenon which has been coined `no solitons without horizons' \cite{Gibbons:1990um, Gibbons:1997cc}. 

This result does not extend to higher dimensions. There are numerous examples of asymptotically flat, stationary spacetimes without a black hole, see e.g.~\cite{Bena:2007kg, Gibbons:2013tqa}. In these examples, the spacetime contains non-trivial 2-cycles, or `bubbles',  supported by magnetic flux. Indeed in Einstein-Maxwell theory, the assumption of trivial topology is enough to rule out the existence of solitons \cite{Shiromizu:2012hb}.   The existence of these  `bubbling' spacetimes is closely tied to the failure of black hole uniqueness in five dimensions. Asymptotically flat bubbling spacetimes containing black holes have been constructed~\cite{Kunduri:2014iga}. Interestingly, this leads to a continuous violation of uniqueness in the class of spherical topology black holes and raises a puzzle for the string theory derivation of black hole entropy~\cite{Horowitz:2017fyg}. 

The known bubbling spacetimes are solutions to Einstein-Maxwell theory coupled to a Chern-Simons term for the gauge field (supergravity).   It is tempting to attribute the existence of soliton spacetimes to the Chern-Simons terms. Indeed, in pure Einstein-Maxwell theory it is easy to see that the electric charge of a soliton $Q\sim \int_{S_\infty^{n-2}} \star F = 0 $ regardless of topology (by Stokes' theorem and the Maxwell equation). However, this does not imply the mass vanishes. 

The presence of 2-cycles in the exterior region lead to additional terms in the Smarr mass formula~\cite{Gibbons:2013tqa, Kunduri:2013vka} (see also~\cite{Haas:2014spa}) and the first law of black hole mechanics \cite{Kunduri:2013vka}. It is straightforward to generalise the mass formula~\cite{Kunduri:2013vka} to $n$-dimensional Einstein-Maxwell theory, yielding the following expression for the ADM mass for a soliton
\be\label{ADM}
M_{\text{ADM}} = \frac{1}{4\pi(n-3)} \int_\Sigma\Theta \wedge F  \; ,
\ee
where $\Sigma$ is a Cauchy surface,  $\xi$ is the stationary Killing field and $\Theta=i_\xi \star F \in H^{n-3}(\Sigma)$ encodes the magnetic field. This raises the possibility of positive energy regular spacetimes with non-trivial topology supported by a magnetic field. The known bubbling solutions also carry angular momentum. This also raises the question of whether angular momentum in necessary or is magnetic flux sufficient to support topology. 

The purpose of this note is to address the above questions. In particular,  we will answer the question: do static bubbling spacetimes exist (with or without a black hole) in Einstein-Maxwell theory?  In fact, Gibbons, Ida, and Shiromizu have previously considered the classification of static spacetimes in this theory~\cite{Gibbons:2002bh, Gibbons:2002av, Gibbons:2002ju}. Under the assumption that there are no magnetic fields, they proved that an asymptotically flat, static spacetime containing a non-extremal black hole must be given by a Reissner-Nordstr\"om solution.  It is easy to see that their proof also excludes soliton spacetimes in this class. However, as discussed above bubbling spacetimes must be supported by magnetic flux.  Therefore, to answer the above question requires us to revisit the classification of static spacetimes and include the possibility of a magnetic field.

Our main results are the following:

\begin{theorem} \label{thm1}Any non-trivial, $n\geq 5$ dimensional, asymptotically flat, static solution of the Einstein-Maxwell equations must contain a black hole. 
\end{theorem}

\begin{theorem} \label{thm2} Any $n\geq 5$ dimensional, asymptotically flat, static solution of the Einstein-Maxwell equations containing a non-extremal black hole must be a Reisser-Nordstr\"om solution.
\end{theorem}

Therefore, we find that static soliton spacetimes do not exist and that the domain of outer communication must be topologically trivial. In other words, static bubbling spacetimes do not exist in Einstein-Maxwell theory.
These results will be proved using the ingenious method originally developed by Bunting and Masood-ul-Alam for four dimensional static spacetimes \cite{static1, static2}, later generalised to include degenerate horizons~\cite{Chrusciel:1998rw} and to higher dimensions~\cite{Gibbons:2002bh, Gibbons:2002av, Gibbons:2002ju}.  
In fact,  the purely magnetic case of our theorem could be deduced from the results of~\cite{Emparan:2010ni} by dualising their electric $p$-form $H_p$ for $p=n-2$. \footnote{There have been earlier attempts at classifying solutions with both electric and magnetic fields of an $(n-2)$-form field strength~\cite{Rogatko:2004iz, Rogatko:2006gg}.} In our analysis we will make no assumptions on the Maxwell field and allow for both electric and magnetic fields. We will now briefly sketch the proof of the above results.

The equations of motion for $n$-dimensional Einstein-Maxwell theory are
\bea\label{FE}
&&\tilde{R}_{\mu\nu} = 2 \left( F_{\mu \rho} F_{\nu}^{~~ \rho} - \frac{1}{2(n-2)} F^2 \tilde{g}_{\mu\nu}  \right) \\
&& \td \tilde{\star} F=0 \label{Max}
\eea 
where $\tilde{R}_{\mu\nu}$ is the curvature of the spacetime $(M, \tilde{g})$ and $F= \td A$ for some {\it locally} defined potential $A$. For any static spacetime we can introduce coordinates so that
\be
\tilde{g} = -V(x)^2 \td t^2 + g_{ij}(x) \td x^i \td x^j  \label{static}
\ee
where $\xi = \partial /\partial t$ is the static Killing field, $g$ is a Riemannian metric on a hypersurface $\Sigma$  of constant $t$ and $V$ is a smooth positive function.  If $\xi$ is strictly timelike in the spacetime then $\Sigma$ is a complete manifold.  If $\xi$ is null anywhere the above coordinate system breaks down at the level set $V=0$. In this case, we extend $\Sigma$ to a manifold with a smooth boundary $\partial \Sigma= \{ V=0 \} $, which could correspond to a event horizon or an ergosurface if the hypersurface is null or timelike respectively. We will assume the non-degeneracy condition $\kappa^2 \equiv (\td V)^2|_{\partial \Sigma} \neq 0$, otherwise our analysis is valid regardless of the nature of the boundary $\partial \Sigma$ (in the case of an event horizon $\kappa$ is of course the surface gravity). In fact, it has been shown that for asymptotically flat spacetimes, the static Killing field must be strictly timelike in the domain of outer communications, i.e. $\xi$ can only become null on an event horizon and there are no ergosurfaces~\cite{Chrusciel:2008rh}.  In this work, we only need to invoke the results of~\cite{Chrusciel:2008rh}  to rule out the possibility of \emph{degenerate} ergosurfaces  (i.e. ones with $\kappa = 0$ as defined above).  Stationary, non-static spacetimes containing such surfaces, also known as `evanescent' ergosurfaces, have been discussed in~\cite{Gibbons:2013tqa, Niehoff:2016gbi, Eperon:2016cdd}. 

We assume that the spacetime is asymptotically flat and that the domain of outer communication is globally hyperbolic. Topological censorship \cite{Friedman:1993ty,Chrusciel:1994tr, Galloway1995} then implies $\Sigma$ is simply connected.
We also assume that the static isometry extends to a symmetry of the Maxwell field.  This implies $\td i_\xi F=0$ and therefore we deduce the existence of a globally defined electric potential $\psi$ so that $i_\xi F = -\td\psi$.  Thus we can write the Maxwell field as
\be
F = \td \psi \wedge \td t + B
\ee
where $i_\xi B=0$, so $B \in H^2(\Sigma)$ encodes the magnetic field.  The field equations \eqref{FE} reduce to the following equations on $\Sigma$,
\bea
&&R_{ij} = \frac{\nabla_i \nabla_j V}{V} - \frac{2 \nabla_i \psi \nabla_j \psi}{V^2} + \frac{2 | \nabla \psi |^2 }{(n-2) V^2} g_{ij}  + 2 \left( B_{i k} B_{j}^{~k} - \frac{1}{2(n-2)} | B |^2  g_{ij} \right) \label{Ricg} \\
&&\nabla^2 V = C^2 \frac{| \nabla \psi |^2 }{V}  + \frac{1}{n-2}  V | B |^2  \; ,   \label{Vpsieq} \qquad \nabla^i  ( V^{-1} \nabla_i \psi)=0, \\
&& \nabla^i (V  B_{ij})=0, \qquad B_{i}^{~j}\nabla_j \psi = 0  \label{Beq} 
\eea
where $R_{ij}, \nabla_i$ and $| \cdot |$ are the Ricci tensor, the metric connection and norm defined by $g_{ij}$, and $C^2 = 2\left(\frac{n-3}{n-2} \right)$. 

The behaviour of the fields near an inner boundary $\partial \Sigma$ may be determined as follows. Given our assumptions we may choose $V$ itself as a coordinate and write
\begin{equation}
g = \rho^2 \td V^2 + h_{\alpha \beta} \td x^\alpha \td x^\beta
\end{equation}
where $x^\alpha$ are coordinates on the $(n-2)$-dimensional level sets of $V$ and $\rho= | \td V |^{-1/2}$ so  $\rho|_{\partial \Sigma} = \kappa^{-1}$. Without using the field equations the spacetime invariant 
\bea
\tilde{R}_{\mu\nu\rho\sigma}\tilde{R}^{\mu\nu\rho\sigma} &=& \frac{4}{V^2} (\nabla_ i \nabla_j V) (\nabla^i  \nabla^j V) + R_{ijkl}R^{ijkl} \nonumber \\
&=& \frac{4}{V^2} \left( \frac{(\partial_V \rho)^2}{\rho^6} + \frac{2 D^\alpha \rho D_\alpha \rho}{\rho^4} +\frac{K_{\alpha \beta} K^{\alpha\beta}}{\rho^2} \right) +  R_{ijkl}R^{ijkl}
\eea
where $D_\alpha$ is the metric connection of $h_{\alpha\beta}$ and $K_{\alpha\beta} = \tfrac{1}{2} \rho^{-1} \partial_V h_{\alpha \beta}$ is the extrinsic curvature of the level sets of $V$. This is a sum of squares and hence requiring smoothness of this invariant as $V \to 0$  implies 
\be
D_\alpha \rho = O(V), \qquad K_{\alpha \beta} = O(V), \qquad \partial_V \rho= O(V)  \label{BC1}
\ee
and hence their limiting values at $V\to 0$ all vanish.  To deduce behaviour of $\psi$ we can use 
\be
\nabla^2 V =  \frac{1}{\rho} K - \frac{\partial_V \rho}{\rho^3}
\ee
where $K = h^{\alpha \beta} K_{\alpha \beta}$ 
and using the field equations \eqref{Vpsieq} it follows that 
\be
D_{\alpha} \psi = O(V), \qquad \partial_V \psi  = O(V)  \label{BC2}
\ee
and hence their limiting values also vanish.  Note that the first condition in \eqref{BC1} and \eqref{BC2} imply that $\kappa$ and the electric potential are constant on $\partial \Sigma$.

We will consider asymptotically flat solutions. This means  in particular that $\Sigma \setminus K$, where $K$ is a compact set, is diffeomorphic to $\mathbb{R}^{n-1} \setminus B(R)$ where $B(R)$ is a ball of radius $R>0$. Furthermore, we require the following asymptotic decay conditions
\bea
g_{ij} &=& \left( 1+ \frac{2 M}{n-3} \frac{1}{r^{n-3}} \right) \delta_{ij} + O(r^{-(n-2)}) \\
V &=& 1- \frac{M}{r^{n-3}} +O(r^{-(n-2)}) \\
\psi &=& \frac{Q}{C r^{n-3}} +O(r^{-(n-2)}) 
\eea
for $r= \sqrt{x^ix^i}>R$, 
where $M$ and $Q$ are constants proportional to the mass $M_{\text{ADM}}$ and electric charge respectively. 
The magnetic field $B$ will also have fall-off conditions, although we will not need this in what follows. In the absence of a black hole, the mass of the spacetime \eqref{ADM} becomes
\begin{equation}
M_{\text{ADM}} = \frac{1}{8\pi(n-3)} \int_\Sigma V \, |B|^2 \; \td \text{vol}(g)   \label{ADMstatic}
\end{equation}
Thus non-trivial soliton spacetimes exist if and only if $B \neq 0$.

We first derive an important consequence of the asymptotic conditions. 
\begin{lemma} For any asymptotically flat, static solution to Einstein-Maxwell theory 
\begin{equation}
M \geq |Q|
\end{equation}
with equality occurring iff $B=0$ 
and $\pm C \psi=1-V$ everywhere on $\Sigma$.
\end{lemma}
\begin{proof} It is convenient to define the smooth functions
\be
F_\pm = V \pm C \psi -1
\ee
The field equations \eqref{Vpsieq} imply that
\be
\nabla^i ( V \nabla_i F_{\pm}) = | \nabla F_{\pm}|^2 + \frac{V^2 | B |^2 }{n-2}
\ee
Integrating this over $\Sigma$ and using the asymptotic conditions, which imply
\be
F_\pm = - \frac{M\mp Q}{r^{n-3}}+ O(r^{-(n-2)})  \; ,  \label{Fpmasymp}
\ee
together with the inner boundary conditions if necessary, yields
\be
M \mp Q = \frac{1}{(n-3)\text{Vol}(S^{n-2})} \int_\Sigma \left[ |\nabla F_\pm |^2+ \frac{1}{n-2} V^2 | B |^2  \right]   \; \td \text{vol}(g)  \geq 0 \; ,
\ee
with equality occurring if and only if $B_{ij}=0$ and either $F_+$ or $F_-$ is a constant (which must vanish by the asymptotic conditions).
\end{proof}

\begin{remark} If $M= |Q|$ the conformal scaling $\hat{g} = \Omega^2 g$ with $\Omega= V^{1/(n-3)}$ implies $\hat{R}_{ij}=0$, $\hat{\nabla}^2 V^{-1}=0$ and that $\hat{g}$ is asymptotically flat with zero mass.  If $(\Sigma, \hat{g})$ is complete then by the positive mass theorem~\cite{Witten:1981mf, SchoenYau} we deduce that $(\Sigma, \hat{g})$ is Euclidean space. In this case the spacetime is a Majumdar-Papapetrou solution and it is natural to expect the general solution in this class must correspond to the multi-centred extreme black hole solutions, although we emphasise this is an open problem.  For dimension $n=4$, this has been proved and requires detailed use of the near-horizon geometry~\cite{Chrusciel:2006pc}. In higher dimensions, more general static near-horizon geometries are possible even with no magnetic field~\cite{Kunduri:2009ud}, which may complicate the classification.  If $\Sigma$ has an inner boundary, the classification remains open (although for $n=4$ this case can be ruled out).
\end{remark}

Henceforth we will  assume $M>|Q|$. In particular,  we will generalize the uniqueness proof of \cite{Gibbons:2002ju} to include the presence of magnetic fields.
\begin{lemma}\label{lemma2} $F_\pm< 0$ on $\Sigma$  if $M>|Q|$.\footnote{For the case of no magnetic field this was implicitly assumed in~\cite{Gibbons:2002ju} (see also~\cite{Rogatko2003}).}
\end{lemma}
\begin{proof}
The field equations  \eqref{Vpsieq} imply
\be
\nabla^2 F_\pm \mp \frac{C \nabla \psi \cdot \nabla F_\pm }{V}  = \frac{1}{n-2} V | B |^2  \geq 0
\ee
Thus if $M>|Q|$  equation (\ref{Fpmasymp}) implies the functions $F_\pm<0$ in the asymptotic end. By the Hopf maximum principle $F_\pm$ cannot attain a maximum in the interior of $\Sigma$. Therefore if $\Sigma$ is complete we must have $F_\pm <0$ everywhere.  If $\Sigma$ has an inner boundary then $n^i \partial_i F_\pm|_{\partial \Sigma}=\kappa>0$, where $n =\kappa^{-1}\td V$ is the unit inward  normal to $\partial \Sigma$ and we have used \eqref{BC2}. Hence $F_\pm$ cannot attain a maximum on the inner boundary and therefore we must again have $F_\pm <0$ everywhere on $\Sigma$. 
\end{proof}

The above lemma  allows us to introduce the conformally related Riemannian manifolds $(\Sigma^\pm,g^\pm)$ with ${g}^{\pm} = \Omega_\pm^2 g$ where 
\be
\Omega_\pm \equiv \left[ \left( \frac{1\pm V}{2} \right)^2 - \frac{C^2 \psi^2}{4} \right]^{1/(n-3)} \;.
\ee
Indeed, positivity of $\Omega_\pm$ follows from the identities $\Omega_+^{n-3}- \Omega_-^{n-3} =  V$ and  $\Omega^{n-3}_- = \tfrac{1}{4} F_+ F_-$. It is easy to check that $(\Sigma^+, g^+)$ is asymptotically flat and has zero ADM mass, and that if $M> |Q|$ the metric $g^-$ extends to the compact manifold $\Sigma^- \cup \{ p \}$ where $p$ is the point at infinity~\cite{Gibbons:2002ju}. A  tedious calculation yields an expression for the scalar curvature of $g^\pm$,
\be \label{scalarcurvature}
\Omega_\pm^2 R^\pm = \frac{| B |^2 }{\Omega_\pm^{n-3}}\left[\Omega_\pm^{n-3}  \mp \frac{ V (1\pm V) }{ (n-3)} \right]+ \frac{1}{ 8 V^2 \Omega_\pm^{2(n-3)}} \left| 2 V \psi \nabla V - (V^2-1 +C^2 \psi^2) \nabla \psi \right|^2
\ee
where we have used \eqref{Ricg} and (\ref{Vpsieq}).
\begin{lemma}  For $n \geq 5$,  $\Sigma$ is conformally flat and
 \be 
B = 0 , \qquad 2 V \psi \nabla V  = (V^2-1 +C^2 \psi^2) \nabla \psi   \label{flat}
\ee
hold everywhere on $\Sigma$. 
\end{lemma}
\begin{proof} Consider the scalar curvature expression \eqref{scalarcurvature}. Only the term proportional to $| B |^2 $ is not manifestly non-negative. 
 The $| B |^2 $ term in $R^-$ is non-negative, and vanishes only if $B=0$, since by Lemma \ref{lemma2} we have $1-V> \pm C \psi$ and hence $1-V > C|\psi| \geq 0$. If $n \geq 5$, the $| B |^2 $ term in $R^+$ is also non-negative, and vanishes only if $B=0$, because
\be
\Omega_+^{n-3}- \frac{V(1+V)}{n-3} = \Omega_-^{n-3} + \frac{V}{n-3} [ n-5+ (1-V)] >0
\ee
Therefore for $n \geq 5$ the conformally rescaled metric has non-negative scalar curvature even when a magnetic field is included.\footnote{This is not true for $n=4$. Indeed, one can have non-zero magnetic fields in four dimensions.}  Thus, if $\Sigma$ is complete,  $(\Sigma^+, g^+)$ is an asymptotically flat complete Riemannian manifold with non-negative scalar curvature and zero mass, so by the positive mass theorem \cite{SchoenYau} it must be $\mathbb{R}^{n-1}$ with $g^+_{ij} = \delta_{ij}$.  

On the other hand, if $\Sigma$ has an inner boundary, we may form a complete manifold $\hat{\Sigma} = \Sigma^+ \cup \Sigma^- \cup \{ p \}$  equipped with metric $\hat{g}$ by pasting $(\Sigma^\pm, g^\pm)$ along the boundaries $\partial \Sigma^\pm = \{ V=0 \}$  and adding the point $p$ at infinity as above. Indeed, since $\Omega_+=\Omega_-$ at $V=0$ the metric is continuous, and the extrinsic curvatures of $\partial\Sigma^\pm$ in $(\Sigma^\pm, g^\pm)$ with respect to the unit inward normal,
\be
K^\pm_{\alpha\beta}= \frac{\kappa}{2(n-3)} \Omega_\pm^{2-n} {h}^\pm_{\alpha\beta}  \; ,  \label{umbilic}
\ee
also match continuously, where  ${h}^\pm_{\alpha\beta}$ is the induced metric and we have used that $\partial\Sigma$ is totally geodesic (\ref{BC1}).  This, together with the fact that $(\hat{\Sigma}, \hat{g})$ is asymptotically flat with zero mass and has non-negative scalar curvature,  is sufficient to invoke the positive mass theorem and conclude $(\hat{\Sigma}, \hat{g}) = (\mathbb{R}^{n-1},\delta)$~\cite{static1, Chrusciel:1998rw, SchoenYau}.   Thus, in either case the scalar curvature vanishes, which implies the conditions (\ref{flat}).
\end{proof}  

\begin{remark} The latter condition in \eqref{flat} implies that level surfaces of $V$ coincide with those of $\psi$.  In fact, it may be directly integrated 
\be 
V^2 = 1 +C^2 \psi^2 - \frac{2M C}{Q} \psi
\ee
where we have fixed the integration constant using the asymptotics.
\end{remark}

We have therefore shown that $B \equiv 0$, i.e. the magnetic field must vanish after all. The analysis then reduces to that of \cite{Gibbons:2002ju}, so we will be brief.  We can write
\be
g_{ij} = \Omega_+^{-2} \delta_{ij}= (v_+ v_-)^{2/(n-3)} \delta_{ij}
\ee
where 
\be
v_\pm = \frac{2}{2 + F_\pm} \; .
\ee A tedious calculation using (\ref{Vpsieq}) and (\ref{flat}) reveals that $v_\pm$ are harmonic functions on $(\Sigma^+, \delta)$.  Lemma 2 implies
that $v_{\pm}>1$ on $\Sigma^+$ and the asymptotic conditions imply
\be
v_\pm = 1+ \frac{M\mp Q}{2 r^{n-3}} +  O(r^{-(n-2)})   \label{asymptv}
\ee
Thus the functions $v_{\pm}$ are bounded on $\Sigma^+$. Therefore, if $\Sigma$ is a complete manifold, we see that  $v_\pm$ are bounded harmonic functions on $\Sigma^+ \cong \mathbb{R}^{n-1}$. Hence they must be constants which coincide with their asymptotic value $v_\pm=1$. It follows that $V=1$ and $\psi=0$ so the solution is just Minkowski spacetime.

If there is an inner boundary, (\ref{umbilic}) shows that the embedding of $\partial \Sigma^+$ must be totally umbilical in $(\mathbb{R}^{n-1}, \delta)$. We conclude that each connected component of $\partial \Sigma^+$ must be a geometric sphere. The possibility of $\partial \Sigma^+$ having multiple connected components may be excluded using an argument given in~\cite{Chrusciel:1998rw} (note this considerably simplifies the proof in~\cite{Gibbons:2002bh, Gibbons:2002av, Gibbons:2002ju}). 
If $\partial \Sigma^+$ has multiple connected components, then $\hat{\Sigma} \setminus \overline{\Sigma}^+\cong \Sigma^{-} \cup \{p\}$ would be a disjoint union of balls in $\mathbb{R}^{n-1}$.  In particular, $\Sigma^-$ could not be connected, in contradiction with $\Sigma^-\cong \Sigma$.  Thus $\partial \Sigma^+$ must have a single connected component which we may identify with the surface $r=r_0$ in $\hat{\Sigma} \cong \mathbb{R}^{n-1}$. Then we have a boundary value problem on ${\Sigma}^+ = \mathbb{R}^{n-1} \setminus B(r_0)$ for the harmonic function $v_\pm$ with the boundary conditions (\ref{asymptv}) and $v_\pm|_{r=r_0}$ is a constant (since it is on a level set of $V$). The unique solution to this is
\be
v_\pm = 1+  \frac{M\mp Q}{2 r^{n-3}} 
\ee
which determines the data $(V,\psi)$. The solution corresponds to the exterior region of the Reissner-Nordstr\"om black hole. This completes the proof of Theorem \ref{thm1} and \ref{thm2}.

\begin{remark}
If there is no black hole the vanishing of the magnetic field and the mass formula (\ref{ADMstatic}) immediately implies $M_{\text{ADM}}=0$ and hence by the positive mass theorem \cite{SchoenYau} the spacetime is Minkowski.  This gives a quick proof of Theorem \ref{thm1}.  
\end{remark}

\begin{remark}
Note that we can also rule out the possibility of any degenerate components of the horizon in the $M > |Q|$ case under consideration.  Any degenerate components of the horizon must correspond to an asymptotic cylindrical end of $\Sigma$ (they lie an infinite proper distance away from any point in $\Sigma$, see e.g. \cite{Chrusciel:1998rw, Khuri:2017fun}).  Therefore, the above gluing procedure can still be carried out to obtain a complete manifold $\hat\Sigma$, now with additional asymptotic ends. An extension of the positive mass theorem may be invoked in this setting~\cite{BartnikChrusciel,Chrusciel:1998rw} (at least for spin manifolds) to conclude again that $\hat\Sigma \cong \mathbb{R}^{n-1}$, which in particular rules out the existence of any such additional ends. 
\end{remark}

It is interesting to consider how these results generalise to Einstein-Maxwell theory coupled to a Chern-Simons term. In particular, for $n=5$ this includes the important case of minimal supergravity. In this case, the only change is  the Maxwell equation (\ref{Max}) acquires a term proportional to $F \wedge F$ (we do not need the precise coefficient).  This leads to the equation for $\psi$ in (\ref{Vpsieq}) acquiring a term $\star_g ( B \wedge B)$, whereas the equation for $V$ is unchanged. In general this will spoil the positivity arguments required for the above results. However, in the purely magnetic case $\psi \equiv 0$, the Maxwell equation implies $B \wedge B=0$ and it is easy to see that the above arguments remain valid. We thus deduce that Theorem \ref{thm1} and \ref{thm2} remain valid for $n=5$ Einstein-Maxwell theory coupled to a CS term if $\psi \equiv 0$.

Finally, it would be interesting to investigate whether stationary, but non-static, asymptotically flat, bubbling  spacetimes exist in higher dimensional Einstein-Maxwell theory. The only known examples are in supergravity which possesses a Chern-Simons term. However, there does not appear to be any reason they cannot occur in pure Einstein-Maxwell theory. As the general mass formula (\ref{ADM}) shows, they would also have to be supported by a magnetic field.

\paragraph{Acknowledgements} HKK is supported by NSERC Discovery Grant 418537-2012. JL is supported by STFC [ST/L000458/1].  HKK thanks the Higgs Centre for Theoretical Physics and the School of Mathematics at the University of Edinburgh for their hospitality and support during the initiation of this work.

\end{document}